\pdfoutput=1
\RequirePackage{ifpdf}
\ifpdf 
\documentclass[pdftex]{sigma}
\else
\documentclass{sigma}
\fi

\newtheorem{Theorem}{Theorem}
\newtheorem{Lemma}[Theorem]{Lemma}
\newtheorem{Proposition}[Theorem]{Proposition}
 { \theoremstyle{definition}
\newtheorem{Definition}[Theorem]{Definition}
 }

\newcommand{\R}{\ensuremath{\mathbb{R}}}
\newcommand{\C}{\ensuremath{\mathbb{C}}}

\newcommand{\dd}{\mathrm{d}}
\newcommand{\rr}{\mathbb{R}}
\newcommand{\zz}{\mathbb{Z}}
\newcommand{\cc}{{\mathbb{C}}}
\newcommand{\e}{\varepsilon}
\newcommand{\I}{{\rm i}}
\newcommand{\ep}{\varepsilon}
\newcommand{\Real}{\ensuremath{\mathrm{Re}}}
\newcommand{\Imag}{\ensuremath{\mathrm{Im}}}

\begin{document}


\newcommand{\arXivNumber}{1509.00701}

\renewcommand{\PaperNumber}{098}

\FirstPageHeading

\ShortArticleName{A Classical Limit of Noumi's $q$-Integral Operator}
\ArticleName{A Classical Limit of Noumi's $\boldsymbol{q}$-Integral Operator}

\Author{Alexei {BORODIN}~$^{\dag^1}$, Ivan {CORWIN}~$^{\dag^2\dag^3\dag^4}$ and Daniel {REMENIK}~$^{\dag^5}$}

\AuthorNameForHeading{A.~Borodin, I.~Corwin and D.~Remenik}

\Address{$^{\dag^1}$~Massachusetts Institute of Technology, Department of Mathematics,\\
\hphantom{${\dag^1}$}~77 Massachusetts Avenue, Cambridge, MA 02139-4307, USA}
\EmailDD{\href{mailto:borodin@math.mit.edu}{borodin@math.mit.edu}}

\Address{$^{\dag^2}$~Columbia University, Department of Mathematics,\\
\hphantom{$^{\dag^2}$}~2990 Broadway, New York, NY 10027, USA}
\EmailDD{\href{mailto:ivan.corwin@gmail.com}{ivan.corwin@gmail.com}}

\Address{$^{\dag^3}$~Clay Mathematics Institute, 10 Memorial Blvd. Suite 902, Providence, RI 02903, USA}
\Address{$^{\dag^4}$~Institut Henri Poincar\'e, 11 Rue Pierre et Marie Curie, 75005 Paris, France}

\Address{$^{\dag^5}$~Departamento de Ingenier\'{i}a Matem\'{a}tica and Centro de Modelamiento Matem\'{a}\-tico,\\
\hphantom{$^{\dag^5}$}~Universidad de Chile, Beauchef 851, Torre Norte, Santiago, Chile}
\EmailDD{\href{mailto:dremenik@dim.uchile.cl}{dremenik@dim.uchile.cl}}

\ArticleDates{Received September 03, 2015, in f\/inal form December 01, 2015; Published online December 03, 2015}

\Abstract{We demonstrate how a known Whittaker function integral identity arises from the $t=0$ and $q\to 1$ limit of the Macdonald polynomial eigenrelation satisf\/ied by Noumi's $q$-integral operator.}


\Classification{05E05; 33D52; 33D52; 82B23}

The {\it class-one $\mathfrak{gl}_{N}$-Whittaker functions} are of great interest in representation theory, integrable systems and number theory (see Gerasimov, Lebedev and Oblezin~\cite{G11} and references therein).
Givental \cite{GKLO, Givental} gave the following integral representation, which we will take as our def\/inition:
\begin{gather*}
\psi_{\lambda}(x)  = \int_{\R^{N(N-1)/2}} \prod_{k=1}^{N-1} \prod_{i=1}^{k}  \dd x_{k,i}\, e^{\mathcal{F}_{\lambda}(X)},
\end{gather*}
where $\lambda = (\lambda_1,\ldots,\lambda_N)\in \C^N$, $x=(x_1,\ldots,x_N)$, $X=(x_{k,i}\colon 1\leq i\leq k\leq N)$, $x_{N,i} = x_i$, and
\begin{gather*}
\mathcal{F}_{\lambda}(X) = \I \sum_{k=1}^{N} \lambda_k \left(\sum_{i=1}^k x_{k,i} - \sum_{i=1}^{k-1} x_{k-1,i}\right) - \sum_{k=1}^{N-1} \sum_{i=1}^{k} \left(e^{x_{k,i}-x_{k+1,i}} + e^{x_{k+1,i+1}-x_{k,i}}\right).
\end{gather*}

Whittaker functions satisfy the following two integral identities.
\begin{Proposition}\label{propbumps}
Suppose $u>0$ and $\lambda,\nu\in\C^N$ with $\Re(\lambda_i+\nu_j)>0$ for all $1\leq i,j\leq N$. Then
\begin{gather*}
\int_{\R^N} \dd x\, e^{-u e^{x_1}} \psi_{-\I\lambda}(x)\psi_{-\I\nu}(x)  = u^{- \sum\limits_{j=1}^{N} (\lambda_j+\nu_j)} \prod_{1\leq i,j\leq N} \Gamma( \lambda_i +  \nu_j),
\end{gather*}
and
\begin{gather*}
\int_{\R^N} \dd x\, e^{-u e^{-x_N}} \psi_{\I\lambda}(x)\psi_{\I\nu}(x)  = u^{- \sum\limits_{j=1}^{N} (\lambda_j+\nu_j)} \prod_{1\leq i,j\leq N} \Gamma( \lambda_i +  \nu_j).
\end{gather*}
\end{Proposition}

The f\/irst of these identities is due to Stade~\cite{Stade}, while the second follows readily by appealing to the fact that $\psi_\lambda(x) = \psi_{-\lambda}(x')$ where $x' = -x_{N-i+1}$. They also follow from \cite[Corollaries~3.6 and~3.7]{OSZ} once one observes that the multiplicative Whittaker functions $\Psi^N_{\lambda}(x)$ in~\cite{OSZ} are related to those def\/ined above via $\Psi^N_{\lambda}(x) = \psi_{\I \lambda}(\log x)$, with $\log x = (\log x_1,\ldots,\log x_N)$.

The Plancherel theory for Whittaker functions (see, for example, \cite[Section~4.1.1]{BigMac}) implies that for $x,y\in\R^N$,
\begin{gather*}
\int_{\R^N}  \dd\lambda\,  \psi_{-\lambda}(x) \psi_{\lambda}(y) m_N(\lambda)= \delta(x-y),
\end{gather*}
where the Sklyanin measure $m_N(\lambda)$ is def\/ined as
\begin{gather*}
    m_N(\lambda)=\frac1{(2\pi)^N N!}\prod_{\substack{i,j=1\\i\neq j}}^N
    \frac1{\Gamma(\I \xi_i- \I \xi_j)},
  \end{gather*}
and $\delta(x-y)$ is the Dirac delta function for $x=y$. Using the above identity one readily observes the equivalence of the results of Proposition~\ref{propbumps} and the following two eigenrelations. The integral operator on the right-hand side of~(\ref{eignwhitone}) is referred to as a {\it dual Baxter operator} by Gerasimov, Lebedev and Oblezin~\cite{G06}, wherein the below result is demonstrated in a rather dif\/ferent manner than that followed in this present work.

\begin{Proposition}\label{coreignwhit}
For $u>0$ and $\Imag(w_i)<0$, $1\leq i\leq N$,
\begin{gather}\label{eignwhitone}
e^{-u e^{-x_1}} \psi_{-w}(x) = \int_{\R^N} \dd\xi\, m_N(\xi)  u^{\I \sum\limits_{k=1}^{N} ( w_i+\xi_i)}\prod_{1\leq i,j\leq N} \Gamma(-\I \xi_i - \I w_j) \psi_{\xi}(x),
\end{gather}
and
\begin{gather}\label{eignwhitN}
e^{-u e^{-x_N}} \psi_{w}(x) = \int_{\R^N}\dd\xi\, m_N(\xi)  u^{\I \sum\limits_{k=1}^{N} ( w_i+\xi_i)}\prod_{1\leq i,j\leq N} \Gamma(-\I \xi_i - \I w_j) \psi_{-\xi}(x).
\end{gather}
\end{Proposition}

The aim of this note is to explain how the eigenrelation in~(\ref{eignwhitN}) arises as the $q\to 1$ limit of a certain eigenrelation involving Noumi's $q$-integral operator and the Macdonald symmetric polynomials.

\begin{Definition}
The $q$-shift operator $T_{q,z_i}$ acts by mapping function $f(z_1,\ldots,z_i,\ldots, z_N)\mapsto f(z_1,\ldots, q z_i,\ldots, z_N)$. For $u\in \C$ of suitably small modulus, {\it Noumi's $q$-integral operator} $\mathfrak{N}^u$ acts on the space of analytic functions in $z_1,\ldots, z_N$ as
\begin{gather*}
\mathfrak{N}^{\zeta}=\sum_{\nu\in(\zz_{\geq0})^N} \zeta^{|\nu|} \prod_{1\leq i<j\leq
  N}\frac{q^{\nu_i}z_i-q^{\nu_j}z_j}{z_i-z_j}
\prod_{i,j=1}^N\frac{(tz_i/z_j;q)_{\nu_i}}{(qz_i/z_j;q)_{\nu_i}}\prod_{i=1}^N (T_{q,z_i})^{\nu_i}.
\end{gather*}
Here $|\nu| = \nu_1+\cdots + \nu_N$.
\end{Definition}

This operator as well as the below eigenrelation it satisf\/ies is due to M.~Noumi. It f\/irst appeared in \cite[Proposition~3.24]{FHHSY}, and further details regarding its derivation are forthcoming in~\cite{NS}. It was subsequently discussed in \cite[Section~4]{BCGS} and a proof of the eigenrelation was given in an appendix therein by E.~Rains. The Macdonald $q$-dif\/ference operators are likewise diagonal in the basis of Macdonald symmetric polynomials with eigenvalues of the form $e_{r}\big(q^{\lambda_1}t^{N-1},q^{\lambda_2}t^{N-2}, \ldots, q^{\lambda_N}\big)$, with $e_r$ the $r^{\rm th}$ elementary symmetric polynomial. While these dif\/ference operators must commute with the Noumi $q$-integral operator there is presently no easy way to express them through each other.

Recall the Macdonald symmetric polynomial $P_{\lambda}(z)$ indexed by partitions $\lambda$ and symmetric in the $z$-variables with coef\/f\/icients which are rational functions of two auxiliary parameters $q,t\in [0,1)$ (see, for example,~\cite{BigMac, M}). We have the following eigenrelation.

\begin{Proposition}
Noumi's $q$-integral operator is diagonal in the basis of Macdonald symmetric polynomials so that for any partition~$\lambda$,
\begin{gather}\label{eqnoumi}
\big(\mathfrak{N}^{\zeta} P_{\lambda}\big)(z_1,\ldots, z_N) = \prod_{1\leq i\leq N} \frac{\big(q^{\lambda_i}t^{N+1-i}\zeta;q\big)_{\infty}}{\big(q^{\lambda_i}t^{N-i}\zeta;q\big)_{\infty}} P_{\lambda}(z_1,\ldots, z_N).
\end{gather}
This eigenrelation can be understood as a formal power series identity in~$\zeta$ or as a convergent series identity, provided~$|\zeta|$ is suitably small.
\end{Proposition}

The remainder of this note will be devoted to showing how when $t=0$ and $q\to 1$ under particular scaling this eigenrelation yields that of~(\ref{eignwhitN}). We will proceed formally. Various uniform estimates would be necessary for this particular route to yield a rigorous derivation of~(\ref{eignwhitN}). Since this identity has already been proved through other means, we do not provide these necessary details.

Our derivation has two steps. The f\/irst takes the termwise limit of the Noumi operator, yielding the eigenrelation~(\ref{eq:Dt}). It is this step which we do not fully justify. The second step uses residue calculus to rewrite the resulting summation as the claimed contour integral formula.

From here on set $t=0$ and take the following scalings
\begin{gather*}
  q=e^{-\ep},\qquad\lambda_k=(N-2k+1)\ep^{-1}\log\big(\ep^{-1}\big)+\ep^{-1}x_k,\qquad z_k=e^{\I\ep w_k},\qquad \zeta=-u\ep^N.
\end{gather*}
Furthermore, def\/ine the following scaled version of $t=0$ Macdonald symmetric polynomial
\begin{gather*}
\psi^\ep_w(x)=\ep^{\frac{N(N-1)}{2}+\frac{N(N-1)}{2}A(\ep)}P_{\lambda}(z),\qquad A(\ep)=-\tfrac16\pi^2\ep^{-1}-\ep^{-1}\log(\ep/2\pi).
\end{gather*}
Note, at $t=0$, the Macdonald symmetric polynomial has been identif\/ied with the class-one $q$-Whittaker functions, as def\/ined in the work of
 Gerasimov, Lebedev and Oblezin~\cite{G08} (the term ``$q$-Whittaker function'' is used to denote dif\/ferent objects by dif\/ferent authors).

Since we will not be proving a theorem, let us summarize here the result we will formally demonstrate. Under the above scaling, assume that a~sequence of symmetric function~$f^{\e}(z)$ have a limit as~$\e\to 0$ to some~$\tilde{f}(w)$. Then, formally we will show that
\begin{gather*}
\big(\mathfrak{N}^{\zeta} f^{\e}\big)(z) \to \big(\tilde{\mathfrak{N}}^{-u} \tilde{f}\big)(w),
\end{gather*}
where the limiting operator $\tilde{\mathfrak{N}}^{u}$ is def\/ined in~(\ref{eq.fasww}) and is identif\/ied with the dual Baxter operator through Lemma~\ref{lem:contourIntId}. This is the key result of this paper.

{\bf Step 1.}  Let us consider how~(\ref{eqnoumi}) scales as $\ep\to 0$. The right-hand side at $t=0$ equals
\begin{gather*}
\frac1{(q^{\lambda_N}\zeta;q)_{\infty}} P_{\lambda}(z_1,\ldots, z_N).
\end{gather*}
Observe that from the convergence of the $e_{q}$-exponential to the usual exponential (see \cite[Section~3.1.1]{BigMac}),
\begin{gather*}
\frac{1}{(\zeta q^{\lambda_N};q)_\infty}\longrightarrow e^{-ue^{-x_N}}.
\end{gather*}
In \cite[Theorem~3.1]{G11}, it was explained how $t=0$ Macdonald polynomials (i.e., $q$-Whittaker functions) converge to Whittaker functions. Certain tail estimates necessary for that argument were further provided in \cite[Theorem~4.1.7]{BigMac}. The resulting convergence holds that
\begin{gather*}
\psi^\ep_w(x)\longrightarrow\psi_w(x).
\end{gather*}
Note that in both~\cite{G11} and~\cite{BigMac} there was a mistake -- the prefactor for $A(\ep)$ in the def\/inition of~$\psi^\ep_w(x)$ should be $\frac{N(N-1)}{2}$ (as above) whereas in \cite[Theorem~3.1]{G11} and \cite[Theorem~4.1.7]{BigMac} it was written as $\frac{(N-1)(N+2)}{2}$. Let us brief\/ly explain where this error came from (in reference to the proof of  \cite[Theorem~4.1.7]{BigMac}). Comparing equation~(3.8) with~(4.25) we f\/ind that in~(4.25) the term $\Delta^{\e}(\underbar{x}_{\ell})$ should actually be  $\Delta^{\e}(\underbar{x}_{\ell+1})$ (here $\ell+1=N$). This error resulted in neglecting $\ell+1$ extra factors of~$e^{A(\ep)}$ which jives with the dif\/ference between the incorrect and correct powers $\frac{(N-1)(N+2)}{2} - \frac{N(N-1)}{2} = N$.

Turning to the left-hand side, let us consider how each term in the summation scales:
\begin{gather*}
\frac{q^{\nu_i}z_i-q^{\nu_j}z_j}{z_i-z_j} =\frac{e^{-\ep\nu_i+\I\ep w_i}-e^{-\ep\nu_j+\I\ep
    w_j}}{e^{\I\ep w_i}-e^{\I\ep w_j}}
\longrightarrow\frac{\I(\nu_j-\nu_i)+(w_j-w_i)}{w_j-w_i},\\
  \frac{1}{(qz_i/z_j;q)_{\nu_i}} =\frac1{1-e^{-\ep+\I\ep(w_i-w_j)}}\frac1{1-e^{-2\ep+\I\ep(w_i-w_j)}}\dotsm\frac1{1-e^{-\nu_i\ep+\I\ep(w_i-w_j)}}\\
\hphantom{\frac{1}{(qz_i/z_j;q)_{\nu_i}}}{}
=\frac1{\ep(1-\I(w_i-w_j))}\dotsm\frac1{\ep(\nu_i-\I(w_i-w_j))}+\rm{l.o.t.}\\
\hphantom{\frac{1}{(qz_i/z_j;q)_{\nu_i}}}{}
=\ep^{-\nu_i}\frac{\Gamma(1+\I(w_j-w_i))}{\Gamma(1+\nu_i+\I(w_j-w_i))}+\rm{l.o.t.},
\end{gather*}
where $\rm{l.o.t.}$ denotes lower order terms in~$\ep$. Also note that
\begin{gather*}
\prod_{i=1}^NT^{\nu_i}_{q,z_i}P_{\lambda}(z)=P_{\lambda}\big(e^{\I\ep w_1-\ep\nu_1},\dots,e^{\I\ep
    w_N-\ep\nu_N}\big)
=\prod_{i=1}^NS^{\nu_i}_{\I,w_i}\psi^\ep_w(x)\ep^{-\frac{N(N-1)}{2}-\frac{(N+2)(N-1)}{2}A(\ep)},
\end{gather*}
where the ef\/fect of $S_{a,w_i}$ is to shift~$w_i$ by~$a$. Putting all of this together we
deduce that
\begin{gather}\label{eq.fasww}
\tilde{\mathfrak{N}}^{u} =   \sum_{\nu\in(\zz_{\geq0})^N} (u)^{|\nu|} \prod_{i<j}\frac{w_j-w_i+\I(\nu_j-\nu_i)}{w_j-w_i}
\prod_{i,j}\frac{\Gamma(1+\I(w_j-w_i))}{\Gamma(1+\nu_i+\I(w_j-w_i))}\prod_{i=1}^N(S_{\I,w_i})^{\nu_i}
\end{gather}
satisf\/ies the eigenrelation
\begin{gather}\label{eq:Dt}
\tilde{\mathfrak{N}}^{-u} \psi_{w}(x) = e^{-ue^{-x_N}}\psi_w(x).
\end{gather}
This completes the f\/irst step of our derivation. As we already mentioned, various estimates are needed in order to turn this into a rigorous proof.

{\bf Step 2.} The purpose of the second step is to show how~$\tilde{\mathfrak{N}}^{-u}$ can be rewritten in terms of contour integrals through residue calculus. We will obtain this as a consequence of the following slightly more general result:

\begin{Lemma}\label{lem:contourIntId}
  Let $a$ be a positive real number and consider a symmetric function~$f$ defined on the set $\{w\in\cc^N \colon \Imag(w_i)\leq-a,\,i=1,\dots,N\}$ which is bounded and analytic in each variable in this set.
  Then
    \begin{gather}\label{eq.fasf}
	\big(\tilde{\mathfrak{N}}^{-u}f\big)(w)
	 =\int_{(a+\I\rr)^N}\dd \xi\, s_N(\xi)\,u^{\sum_i(\I w_i-\xi_i)}\prod_{i,j=1}^N\Gamma(\xi_j-\I
  	 w_i) f(-\I\xi)
	\end{gather}
	for $u>0$, where~$s_N$ is the following variant of the Sklyanin measure:
	\begin{gather*}
	s_N(\xi_1,\dots,\xi_N)=\frac1{(2\pi\I)^NN!}\prod_{\substack{i,j=1\\i\neq j}}^N
	\frac1{\Gamma(\xi_i-\xi_j)}.
	\end{gather*}
\end{Lemma}

To see that this implies the desired result, let $f(w)=\psi_w(x)$ and observe that~$f$ satisf\/ies the necessary boundedness thanks to \cite[Lemma~4.1.19]{BigMac}.
Making the change of variables $\xi \mapsto -\I\xi$ and combining~(\ref{eq.fasf}) with~(\ref{eq:Dt}) we easily deduce~(\ref{eignwhitN}).

\begin{proof}[Proof of Lemma~\ref{lem:contourIntId}]
We will prove the result by expanding the integral on the right-hand side of~(\ref{eq.fasf}) into residues and matching the result with the summation on the right-hand side of~(\ref{eq.fasww}).

For $i,j=1,\dots,N$, the integrand on the right-hand side of~(\ref{eq.fasf}) has singularities coming from the factor $\Gamma(\xi_j-\I w_i)$ at each point of the form $\xi_j=\I w_i-\nu$ for $\nu\in\zz_{\geq0}$ (where $\zz_{\geq0}$ stands for the set of non-negative integers). By shifting each contour to the left to $-\infty$ we will pick up each of these singularities, and the result is that the integral will equal the sum of the associated residues (The fact that the contours of integration can be deformed in this way follows from our assumption on $f$ and the known asymptotics of the Gamma function, which give $c_1  e^{-\frac{\pi}2|\Imag(z)|}|\Imag(z)|^\eta\leq|\Gamma(z)|\leq c_2 e^{-\frac{\pi}2|\Imag(z)|}|\Imag(z)|^\eta$ for $\Real(z)$ in some f\/inite interval and some $c_1,c_2>0$ and $\eta\in\rr$, see, e.g., \cite[(6.1.45)]{abrSteg};
observe also that, thanks to the factor $u^{\sum_i(\I w_i-\xi_i)}$, there is no pole at $-\infty$).
Each of these residues is evaluated at a point of the form $(\xi_1,\dots,\xi_N)=(\I w_{m_1}-\nu_{1},\dots,\I w_{m_N}-\nu_{N})$ for some $(m_1,\dots,m_N)\in\{1,\dots,N\}^N$ and $\nu_{i}\geq0$ for $i=1,\dots,N$, so the right-hand side of~(\ref{eq.fasf}) equals
  \begin{gather*}
    \frac1{N!}\sum_{(m_1,\dots,m_N)\in\{1,\dots,N\}^N}\sum_{\nu\in(\zz_{\geq0})^N}
    u^{\sum_i(\I w_i-\I w_{m_i}+\nu_i)}\frac{(-1)^{\sum_i\nu_i}}{\prod_i\nu_i!}f(w_{m_1}+\I\nu_1,\dots,w_{m_N}+\I\nu_N)\\
  \qquad{}\times  \prod_{i\neq j}\Gamma(\I w_{m_j}-\nu_j-\I
      w_i)\prod_{i\neq j}\frac1{\Gamma(\I w_{m_i}-\nu_i-\I w_{m_j}+\nu_j)}.
  \end{gather*}

If $m_i=m_j$ for some $i\neq j$ then the associated term in the above sum equals 0. In fact, such a term has a factor of the form $\Gamma(\nu_j-\nu_i)^{-1}\Gamma(\nu_i-\nu_j)^{-1}$, which equals 0 because $\nu_j-\nu_i\in\zz$. Hence the above sum is restricted to the case where $m_i=\sigma(i)$ for some $\sigma\in S_N$ and equals
  \begin{gather*}
    \frac1{N!}\sum_{\sigma\in S_N}\sum_{\nu\in(\zz_{\geq0})^N}
    u^{\sum_i(\I w_i-\I w_{\sigma(i)}+\nu_i)}\frac{(-1)^{\sum_i\nu_i}}{\prod_i\nu_i!}
    f(w_{\sigma(1)}+\I\nu_1,\dots,w_{\sigma(N)}+\I\nu_N)\\
   \qquad{}\times \frac{\prod\limits_{i=1}^N\prod\limits_{j\neq i}
      \Gamma(\I w_{\sigma(j)}-\nu_j-\I w_i))}{\prod\limits_{i<j}\Gamma(\I w_{\sigma(i)}-\nu_i-\I w_{\sigma(j)}+\nu_j))
      \Gamma(\I w_{\sigma(j)}-\nu_j-\I w_{\sigma(i)}+\nu_i))}.
  \end{gather*}
  By the symmetry of~$f$, the sum in~$\nu$ does not depend on
  $\sigma\in S_N$, and therefore if we choose~$\sigma$ to be the identity we deduce that
  the sum equals
  \[\sum_{\nu\in(\zz_{\geq0})^N}(-u)^{\sum_i\nu_i}\frac{1}{\prod_i\nu_i!}f(w+\I\nu)
    \frac{\prod\limits_{i\neq j}\Gamma(\I(w_{j}-w_{i})-\nu_j))}
    {\prod\limits_{i<j}\Gamma(\I(w_{i}-w_j)-\nu_{i}+\nu_j)\Gamma(\I(w_j-w_i)-\nu_j+\nu_i)}.\]
  To complete the argument matching the right-hand side of~(\ref{eq.fasf}) to the right-hand side of~(\ref{eq.fasww}) it suf\/f\/ices to show that
  \begin{gather*}
    \frac{1}{\prod_i\nu_i!}\frac{\prod\limits_{i\neq j}\Gamma(\I(w_{j}-w_{i})-\nu_j))}
    {\prod\limits_{i<j}\Gamma(\I(w_{i}-w_j)-\nu_{i}+\nu_j)\Gamma(\I(w_j-w_i)-\nu_j+\nu_i)}\\
   \qquad{} =\prod_{i<j}\frac{w_j-w_i+\I(\nu_j-\nu_i)}{w_j-w_i}
    \prod\limits_{i,j}\frac{\Gamma(1+\I(w_j-w_i))}{\Gamma(1+\nu_i+\I(w_j-w_i))}.
  \end{gather*}
Observe f\/irst of all that in the last product on the right-hand side of the above equation, the terms $i=j$ equal
  $\Gamma(1+\nu_i)^{-1}=1/(\nu_i!)$. Therefore if we write $r_i=\I w_i$, we need to show that
\begin{gather}
    \label{eq:gammaIdentity2}
    \prod_{i\neq j}\frac{\Gamma(r_{j}-r_{i}-\nu_j))}
    {\Gamma(r_{i}-r_j-\nu_{i}+\nu_j)}
    =\prod_{i<j}\frac{r_j-r_i-\nu_j+\nu_i}{r_j-r_i}
    \frac{\Gamma(1+r_i-r_j)\Gamma(1+r_j-r_i)}{\Gamma(1+\nu_j+r_i-r_j)\Gamma(1+\nu_i+r_j-r_i)}.
  \end{gather}
  Recalling Euler's ref\/lection formula
  \begin{gather}
    \label{eq:euler}
    \Gamma(1-z)\Gamma(z)=\frac\pi{\sin(\pi z)}
  \end{gather}
  and the fact that $z\Gamma(z)=\Gamma(1+z)$ we have
\begin{gather*}
\frac{\Gamma(1+r_i-r_j)\Gamma(1+r_j-r_i)}{r_j-r_i}=\frac\pi{\sin(\pi(r_j-r_i))}
\end{gather*}
  and
  \begin{gather*}
    \frac{r_j-r_i-\nu_j+\nu_i}{\Gamma(1+\nu_j+r_i-r_j)\Gamma(1+\nu_i+r_j-r_i)}
    \\
    \qquad{}=\frac{\sin(\pi(r_j-r_i-\nu_j+\nu_i))}\pi\frac{\Gamma(1-r_j+r_i+\nu_j-\nu_i)}{\Gamma(1+\nu_j+r_i-r_j)}
    \frac{\Gamma(1+r_j-r_i-\nu_j+\nu_i)}{\Gamma(1+\nu_i+r_j-r_i)}.
  \end{gather*}
  Using these identities and the fact that $\sin(\pi(a+k))=(-1)^k\sin(\pi a)$, the right-hand side of~\eqref{eq:gammaIdentity2} becomes
  \begin{gather*}
  \prod_{i<j}(-1)^{\nu_i-\nu_j}\prod_{i\neq
    j}\frac{\Gamma(1-r_j+r_i+\nu_j-\nu_i)}{\Gamma(1+\nu_j+r_i-r_j)},
  \end{gather*}
  and hence~\eqref{eq:gammaIdentity2} is equivalent to
  \begin{gather*}
  \prod_{i<j}(-1)^{\nu_i-\nu_j}\prod_{i\neq
    j}\frac{\Gamma(1-r_i+r_j+\nu_i-\nu_j)}{\Gamma(1+\nu_j+r_i-r_j)}\frac{\Gamma(r_{i}-r_j-\nu_{i}+\nu_j)}
  {\Gamma(r_{j}-r_{i}-\nu_j)}=1.
  \end{gather*}
  But using \eqref{eq:euler} again the left-hand side equals
  \begin{gather*}
  \prod_{i<j}(-1)^{\nu_i-\nu_j}\prod_{i\neq
    j}\frac{\pi}{\sin(\pi(r_i-r_j-\nu_i+\nu_j))}\frac{\sin(\pi(r_j-r_i-\nu_j))}{\pi}
  =(-1)^\kappa
  \end{gather*}
  with
  \begin{gather*}
  \kappa=\sum_{i<j}(\nu_i-\nu_j)+\sum_{i\neq j}(\nu_i-2\nu_j)\\
  \hphantom{\kappa}{} =
  \sum_{m=1}^n(n-2m+1)\nu_m-(n-1)\sum_{m=1}^n\nu_m=2\sum_{m=1}^n(1-m)\nu_m,
\end{gather*}
  which f\/inishes our derivation since~$\kappa$ is even.
\end{proof}

\subsection*{Acknowledgements}
We appreciate helpful comments from our referees. AB was partially supported by the NSF grant DMS-1056390. IC was partially supported by the NSF through DMS-1208998 as well as by the Clay Mathematics Institute through the Clay Research Fellowship, by the Institute Henri Poincar\'e through the Poincar\'e Chair, and by the Packard Foundation through a Packard Foundation Fellowship. DR was partially supported by Fondecyt Grant 1120309, by Conicyt Basal-CMM, and by Programa Iniciativa Cient\'if\/ica Milenio grant number NC130062 through Nucleus Millenium Stochastic Models of Complex and Disordered Systems.

\pdfbookmark[1]{References}{ref}
\LastPageEnding


\begin{thebibliography}{99}
\bibitem{NS}
Noumi M., Sano A., An inf\/inite family of higher-order dif\/ference operators that
  commute with {R}uijsenaars operators of type~{A}, {i}n preparation.
\end{thebibliography}
\end{document}